\documentclass[a4paper]{article}

\usepackage{amsfonts}
\usepackage{amssymb}
\usepackage{amsmath}
\usepackage{amsthm}
\usepackage{amsmath}

\usepackage{mdframed}

\newtheorem{definition}{Definition}
\newtheorem{remark}{Remark}
\newtheorem{theorem}{Theorem}
\newtheorem{proposition}{Proposition}

\newtheorem{example}{Example}
\newtheorem{assumption}{Assumption}

\begin{document}

\title{The Incentive Ratio in Exchange Economies}

\author{Ido Polak\thanks{I am grateful to Xiaohui Bei and Satoru Takahashi for discussions on the topic and many useful  suggestions and comments on earlier versions of this manuscript. All remaining errors are of course my own.}\\
Division of Mathematical Sciences\\
School of Physical and Mathematical Sciences \\
 Nanyang Technological University \\
21 Nanyang Link \\
Singapore 637371 \\
S120059@e.ntu.edu.sg
}

\maketitle

\begin{abstract}
The incentive ratio measures the utility gains from strategic behaviour. Without any restrictions on the  setup, ratios for linear, Leontief and Cobb--Douglas exchange markets are unbounded, showing that manipulating the equilibrium is a worthwhile endeavour, even if it is computationally challenging. Such unbounded improvements can be achieved even if agents only misreport their utility functions. This provides a sharp contrast with previous results from Fisher markets. When the Cobb--Douglas setup is more restrictive, the maximum utility gain is bounded by the number of commodities. By means of an example, we show that it is possible to exceed a known upper bound for Fisher markets in exchange economies.
\end{abstract}


\section{Introduction}

General equilibrium theory and (noncooperative) game theory are among the most succesful and well-studied areas in economic theory. The former seeks to explain the existence of equilibria in multiple markets at the same time. The latter serves as the primary tool for predicting, analysing and describing the behaviour of rational agents' actions both in and out of equilibrium. 

Both branches have yielded an abundance of literature and served as the basis of much fruitful research. In this paper, we try to combine the two approaches for exchange economies. Specifically, we ask how much any individual agent can gain from strategically misreporting some of his/her utility function. The primary tool for this will be the \textit{incentive ratio} as introduced by \cite{Chen2011}, \cite{Chen2012}. In a nutshell, when all agents report truthfully, an equilibrium will materialise, assigning all agents an equilibrium bundle. But given the endowments and utility functions of the other agents, a strategic agent may be better off by misreporting his/her own characteristics, thereby enforcing another equilibrium and, therefore, potentially obtaining a better equilibrium bundle. The incentive ratio tries to capture the maximal magnitude of such gains over all envisionable setups. A formal definition will follow below. The results presented here suggest that, contrary to previous findings in Fisher markets, the gains from strategic behaviour may be significant, even allowing an agent to improve his equilibrium utility without bound. If we impose the (common) restriction that all agents possess at least a little bit of every commodity and the market is strongly competitive, the utility gain in Cobb--Douglas markets is bounded by the number of commodities, but it may exceed the upper bound from Fisher markets, which we show by means of an example. The results obtained show a sharp contrast with the findings in \cite{Chen2011}, \cite{Chen2012}: in the Fisher market setup, incentive ratios are bounded by the small constants 2, 2 and $e^{1/e} \approx 1.44$ for linear, Leontief and Cobb--Douglas markets, respectively. 

\subsection*{Related work} 
 In \cite{AAAI148373}, price of anarchy bounds are computed for linear, Leontief and Cobb--Douglas markets in the Fisher model. That is, the ratio between (utilitarian) welfare of all agents in the worst possible pure Nash equilibrium (for which existence is proven) and maximum welfare. It primarily differs from the analysis presented here in that it focuses on welfare of all agents rather than measuring the benefits of strategic behaviour for one specific agent. As in this paper, the strategic variable of interest is the utility function. Arguably, misreporting the utility function i.e. one's preferences, is easier than misreporting (physical) endowments, which can, at least in theory, be inspected and whose withholding may be difficult and costly (e.g. due to storage costs). It is known that the Walrasian mechanism is susceptible to manipulation via endowments: via withholding endowments and recovering it fully \cite{Postlewaite}, recovering part of it \cite{Thomson} (see also \cite{Yi}) and even destroying part of one's initial endowment \cite{AumannPeleg}. However, the aforementioned studies are qualitative, that is, they show that manipulation of the equilibrium mechanism is possible, but do not quantify it. The incentive ratio is a first step to quantifying the possible gains of misreporting in exchange economies. It is not so difficult to see from Proposition \ref{prop:leontief} below that even when destroying part of one's initial endowment (in the case $e_i$ is the strategic variable of interest), the incentive ratio may still tend to infinity.

The idea that an agent may act strategically in a (Fisher) market by misreporting his/her utility function in order to get a better equilibrium bundle, compared to the scenario where everyone is truthful, was already considered in \cite{Adsul2010} for the case of linear utility functions. In this paper, Nash equilibria and the relation with conflict-freeness (meaning that every agent can get an optimal equilibrium bundle given the strategy profile) is studied. It is shown that being conflict-free is a necessary condition for a Nash equilibrium, and for a symmetric strategy profile it is also sufficient. 
The incentive ratio was first coined in \cite{Chen2011}, where the strategic variable of interest is the (Leontief) utility function of a player and bidding the true budget is a dominant strategy. In \cite{Chen2012}, \cite{AAAI1612300}, a slightly more sophisticated version of the incentive ratio is presented, in which players may also strategise on their endowments. The ``exchange market game'' is introduced in \cite{Mehta2013}, and agents have linear utility functions. They may lie about their utility function to manipulate the outcome of the exchange process. It is consequently shown that a symmetric strategy profile is a  Nash equilibrium if and only if it is conflict-free. Several properties of the symmetric Nash equilibria (e.g. Pareto-optimal payoffs, a characterisation for its uniqueness) are subsequently derived. 

\subsection*{Organisation}
The rest of this paper is organised as follows. Section \ref{sec:prelim} discusses the necessary machinery, definitions and introduces some notation. Next, Section \ref{sec:results} presents the results for incentive ratios in Linear, Leontief and Cobb--Douglas exchange economies. The latter receives most attention. It is shown that, even if the strategy space is restricted to the utility function only, all three incentive ratios are unbounded. Under some specific, common assumptions, the $e^{1/e}$ bound is recovered for Cobb--Douglas markets with 2 commodities. For the case where there are at least 3 commodities, the incentive ratio is shown to be greater than $e^{1/e}$ but bounded by the number of commodities.  Finally, \ref{sec:conclusion} concludes and provides some directions for future research.

\section{Preliminaries, exchange economies and competitive equilibrium}
\label{sec:prelim}

We use the following notation. Suppose $x, y \in \mathbb{R}^n$. Then $ x \cdot y = \sum_{k = 1}^n x_k y_k$ denotes the dot product of $x$ and $y$. $x \leq y$ means $x_k \leq y_k$ for $k = 1, \dots, n$. 
 For a vector $u = (u_1, \dots, u_n)$, by $u_{-i}$ we mean the vector $(u_1, \dots, u_{i-1}, u_{i+1}, \dots, u_n)$ (i.e. all entries except the $i$-th). 
For positive integer $n$, we use $[n]$ as shorthand notation for the set $\{1, \dots, n\}$. $I_m$ is the $m \times m$ identity matrix. The transpose of a matrix $M$ is denoted by $M^T$ and its determinant by $\left| M \right|$. If $f : \mathbb{R}^m \rightarrow \mathbb{R}^n$, then $Df(x)$ represents the Jacobian matrix of $f$ at $x$.

Depending on agents' endowments we speak either of Fisher markets \cite{RePEc:cwl:cwldpp:1272} (in case of monetary endowments) or, more generally, in case of (a vector of) commodity endowments, of exchange economies. In the former case, all agents bring a certain amount of cash to the market and budgets are exogeneous to the model, whereas in the latter the value of the endowment is endogeneous, determined by market prices. That is,  in exchange economies, each agent $i$ is endowed with a bundle of commodities $e_i \in \mathbb{R}_+^m$, whereas in the case of Fisher markets, he simply possesses an amount of wealth $w_i \in \mathbb{R}_+$. We will focus on the more general model of exchange economies with $n$ agents and $m$ commodities. 

\begin{definition}[Exchange economy] An exchange economy (henceforth also simply referred to as economy) is a tuple $\xi = ((u_i)_{i = 1} ^n, (e_i)_{i = 1}^n)$, where \\ $u_i: \mathbb{R}_{+}^m \rightarrow \mathbb{R}$ is the utility function of agent $i \in [n]$ and $e_i \in \mathbb{R}_{+}^m$ is a vector where $e_{ij}$ indicates how much agent $i \in [n]$ possesses of commodity $j \in [m]$.
\end{definition}

In an economy, agents interact with each other to obtain a bundle $x_i \in \mathbb{R}_+^m$ by trading commodities given a price vector $p \in \mathbb{R}^m$. If $p$ is such a price vector, then every agent solves the following consumer problem \eqref{prob:cp}.

\begin{definition}[Demand]
\begin{equation}
\begin{aligned}
& {\text{maximize}}
& & u_i(x_i) \\
& \text{subject to}
& & p \cdot x_i   \leq  p \cdot e_i \\
& & & x_i   \geq   0
\end{aligned}
\tag{$\mathcal{CP}$}
\label{prob:cp}
\end{equation}
We call the set of solutions to problem \eqref{prob:cp} the demand of agent $i$ (at prices $p$). 
\end{definition}
We can write $x_i(p, p \cdot e_i)$ to show explicitly that demand depends on endowments and prices. Since prices are in turn  determined by endowments and utility functions, we may also write 
$x_i(u_i, u_{-i}, e) $ or, when it is understood that $u_{-i}$ and $e$ are fixed, simply as 
$x_i(u_i) $.
The cental topic in this study is the notion of equilibrium, specifically, that of Walrasian or competitive equilibrium.

\begin{definition}[Competitive/Walrasian equilibrium]
A competitive equilibrium is a pair $(p,x) \in \mathbb{R}^m \times (\mathbb{R}_{+}^m)^n$ such that:
\begin{enumerate}
\item For all $j \in [m], \sum_{i = 1}^n x_{ij} = \sum_{i = 1}^n e_{ij}$ i.e. markets clear
\item For all $i \in [n], \: x_i$ is a solution to \eqref{prob:cp}, i.e. $x_i$ is the best bundle among the possible choices in the budget set.
\end{enumerate}
\end{definition}

\subsection{Definition of the incentive ratio and illustration in linear markets}
\label{sec:ir}

Every agent is characterized by two parameters, his endowment $e_i$ and his utility function $u_i$. Generally, different endowments and different utility functions will lead to different equilibria. What if an agent purposely misreports his utility function, thereby trying to get a better equilibrium allocation? 

The incentive ratio is a concept introduced in \cite{Chen2011}. It attempts to measure the (maximum) benefits of manipulating the equilibrium mechanism by strategically misreporting personal parameters. Formally, we define it as follows (adapted for exchange economies, the original definition was given for Fisher markets  in \cite{AAAI1612300}, \cite{Chen2011}, \cite{Chen2012}):
\newpage
\begin{definition}[Incentive ratio]
The incentive ratio of agent $i$ in a market $M$ (e.g. linear, Cobb--Douglas or Leontief), denoted $\zeta_i^M$, is defined as:
$$\zeta_i^M = \max_{u_{-i} \in U_{-i}, e_{-i} \in (\mathbb{R}_+^{m})^{n-1}} \max_{u'_i \in U_i}  \frac{\max_{x' \in \mathcal{E}(u'_i)} u_i(x'_i(u'_i, u_{-i}, e))}{\min_{x \in \mathcal{E}(u_i)} u_i(x_i(u_i, u_{-i}, e))}$$

The incentive ratio of the market $M$ is subsequently defined as
$$\zeta^M = \max_{i \in [n]} \zeta_i^M$$
\end{definition}

\begin{remark}
In this definition:
\begin{itemize}
\item Variables with a prime $(')$ refer to the scenario in which agent $i$ misreports his parameters (and all other agents report truthfully). That is, he reports $u'_i$ and as a result, obtains a bundle $x'_i(u'_i)$. Notice that this bundle is evaluated by the true utility function.
\item Given that player $i$ reports $\tilde{u}_i$ (i.e. truthful or not) as his utility function (and the other players $u_{-i}$), we denote by $\mathcal{E}( \tilde{u}_i)$ the set of equilibrium allocations, that is,
$$\mathcal{E}(\tilde{u}_i) = \{x \in (\mathbb{R}_+^m)^n \vert \: \exists p \in \mathbb{R}_+^m \:\: (p,x) \text{ is a Walras equilibrium}\}.$$ 
under some (mild) assumptions this set is nonempty, but it could contain multiple equilibrium allocations.
\item $U_i$ contains the admissible strategies/utility functions for player $i$, including the one that agent he chooses when he misreports his utility function. We denote $U_{-i} = \prod_{k \neq i} U_k$. We will only consider the case where all $U_k$'s are equal to a common $U$, thus $U_{-i} = U^{n-1}$. 

\item From the preceding arguments, we may restrict attention to agent $i$, since all agents can be treated symmetrically, thus we may rewrite the incentive ratio for the market $M$ as $$\zeta^M = \max_{u_{-i} \in U^{n-1}, e_{-i} \in (\mathbb{R}_+^{m})^{n-1}} \max_{u'_i \in U}  \frac{\max_{x' \in \mathcal{E}(u'_i)} u_i(x'_i(u'_i, u_{-i}, e))}{\min_{x \in \mathcal{E}(u_i)} u_i(x_i(u_i, u_{-i}, e))}$$
\end{itemize}

\end{remark}

\begin{example}
The following example for linear markets shows that the consequences for the incentive ratio stemming from the nonuniqueness of equilibrium can be large. Here, the strategic influence of an agent is perhaps less relevant than the equilibrium selection problem that is embodied in the definition of the incentive ratio: the truthful equilibrium could be very bad for an agent (yielding low utility) while the nontruthful equilibrium could be very good for him (yielding high utility). \\ 

Setup:
$e_1 = (\epsilon, 1-\epsilon), $
$e_2 = (1-\epsilon, \epsilon),$
$u_2(x_2) \equiv 0, $
$\epsilon > 0$ and small
\begin{align}
& \textbf{Truthful} & & \textbf{Nontruthful} \nonumber \\ 
& u_1(x_1)  =  x_{11} & & u'_1(x'_1)   =   x'_{11} \nonumber \\
&p  =  (1, 0) && p'   =   (1,1)  \nonumber \\
&x_1  = ( \epsilon, 1-\epsilon), x_2 =(1-\epsilon, \epsilon) && x'_1   =   (1, 0),  x'_2  =   (0,1)\nonumber
\end{align}

We have $$\frac{u_1(x'_1)}{u_1(x_1)} = \frac{1}{\epsilon}$$
Letting $\epsilon$ tend to 0, the incentive ratio tends to $\infty$. 



\end{example}





Although arguably the utility function for agent 2 is not very realistic, he doesn't care about any of the products on the market, it does not seem that such a scenario was ruled out by the authors in \cite{Chen2012}. 
In fact, in \cite{Chen2011}, \cite{Chen2012}, the following markets are considered, with arbitrarily many agents and commodities: 
\begin{enumerate}
\item Linear, i.e. $U$ is the set of utility functions of the form $u(x) = \alpha \cdot x$ where $\alpha \in \mathbb{R}_+^m$.
\item Leontief, i.e. $U$ is the set of utility functions of the form $u(x) = \min_{j \in [m]} \{x_{j}/\alpha_{j} \}$ where $\alpha \in \mathbb{R}_{++}^m$
\item Cobb--Douglas, i.e. $U$ is the set of utility functions of the form $u(x) = \prod_{j = 1}^m x_{j}^{\alpha_{j}}$ where $0 \leq \alpha_{j} \leq 1$ for all $j \in [m]$ and $\sum_{j = 1}^m \alpha_{j} = 1$, so that $U$ is the set of Cobb--Douglas functions that are homogeneous of degree 1.
\end{enumerate}

The (tight) bounds proven in \cite{Chen2011}, \cite{Chen2012} for Fisher markets are 2, 2 and $e^{1/e}$ for linear, Leontief and Cobb--Douglas respectively.

\section{The incentive ratio in Leontief and Cobb--Douglas markets}
\label{sec:results}

We will, without loss of generality, restrict ourselves to scenarios where $e_i \in [0,1]^m$ for all $i \in [n]$ and $\sum_{i = 1}^n e_{ij} = 1$ for all $j \in [m]$. The results indicate that, without any further restrictions on the setup presented in \cite{Chen2011}, \cite{Chen2012} incentive ratios in Leontief and Cobb--Douglas exchange economies are unbounded.

\begin{proposition}
\label{prop:leontief}
The incentive ratio for Leontief and Cobb--Douglas exchange economies equals $+\infty$ (i.e. $\forall n \in \mathbb{N}$ there exists a market such that the incentive ratio for player $i$ is at least $n$).
\end{proposition}

\begin{proof}[Proof. (Leontief)]
$e_1 = (1-\epsilon, \epsilon), $
$e_2 = (\epsilon, 1-\epsilon), $
$u_2(x_2) = \min\{x_{21},x_{22}\}, $
$\epsilon > 0$ and small
\begin{align}
& \textbf{Truthful} & & \textbf{Nontruthful} \nonumber \\ 
& u_1(x_1) = \min\{x_{11},x_{12}\} & & u'_1(x'_1)  = \min\{x'_{11},x'_{12}\}\nonumber \\
& \text{then} & & \text{then} \nonumber \\ 
&p  =  (\delta, 1) && p'   =   (1,1)  \nonumber \\
&x_1  = ((\epsilon+\delta-\delta\epsilon)/(1+\delta), (\epsilon+\delta-\delta\epsilon)/(1+\delta)) && x'_1   =   (1/2, 1/2) \nonumber \\
&x_2 =(1-\epsilon+\delta\epsilon)/(1+\delta), (1-\epsilon+\delta\epsilon)/(1+\delta)) && x'_2  =   (1/2,1/2) \nonumber
\end{align}
We have $$\frac{u_1(x'_1)}{u_1(x_1)} = \frac{1+\delta}{2(\epsilon+\delta-\delta\epsilon)}$$
Letting $\delta, \epsilon$ tend to 0, the incentive ratio tends to $\infty$. 
\end{proof}

\label{sec:cd}

\begin{proof}[Proof. (Cobb--Douglas)]

$e_1 = (1, 0), $
$e_2 = (0, 1), $
$u_2(x_2) = x_{21}^{\epsilon}x_{22}^{(1-\epsilon)}, $
$\epsilon > 0$ and small
\begin{align}
& \textbf{Truthful} & & \textbf{Nontruthful} \nonumber \\ 
& u_1(x_1)  =  x_{11}^{.5}x_{12}^{.5} & & u'_1(x'_1)   =   x'_{11} \nonumber \\
& \text{then} & & \text{then} \nonumber \\ 
&p  =  (1, 1/(2\epsilon)) && p'   =   (1,0)  \nonumber \\
&x_1  = (1/2 , \epsilon) && x'_1   =   (1, 1) \nonumber \\
&x_2 =(1/2, 1-\epsilon) && x'_2  =   (0,0) \nonumber
\end{align}

We have $$\frac{u_1(x'_1)}{u_1(x_1)} = \sqrt{2{\epsilon}^{-1}}$$
Letting $\epsilon$ tend to 0, the incentive ratio tends to $\infty$. \qedhere

\end{proof}

The intuition behind the proof for Cobb--Douglas markets is that it is possible for player 1 to completely annihilate the equilibrium value of player 2's endowment. Since $u_2(e_2) = 0$, the bundle $x_2 = (0,0)$ solves the consumer problem \eqref{prob:cp}, since the indifference curve $u_2(x_2) = 0$ is odd-shaped compared to the other indifference curves.


We make the following assumption to ensure all equilibrium prices are positive; this is rather standard in algorithmic game theory.\footnote{Alternatively, we could assume the existence of a nonmanipulating agent who possesses at least a little bit of all commodities and who desires every commodity.}

\begin{assumption}
\label{assumption1}
\begin{itemize}
\item (Positivity of endowments) Every agent possesses a strictly positive amount of every commodity: $\forall i \in [n], \forall j \in [m] \:\:\: e_{ij} > 0$;
\item (Strong competitiveness (see e.g. \cite{AAAI148373})) Every commodity is demanded by at least one agent: $\forall j \in [m] \: \exists i \in [n] \: \alpha_{ij} > 0$, $\forall j \in [m] \: \exists i \in [n]  \: \alpha'_{ij} > 0$ 
\end{itemize}
\end{assumption}
This entails that the economy excess demand function

$$z(p) := \sum_{i=1}^n (x_i(p, p \cdot e_i) - e_i) = \sum_{i=1}^n x_i(p, p \cdot e_i) - 1 $$ 

 has the gross substitute property and this implies that the equilibrium price is unique (see e.g. \cite{MWG}).


For Cobb--Douglas markets with 2 commodities, the incentive ratio is $e^{1/e}$ i.e. as in the Fisher market scenario. Here, we rule out a setup in which there exists an agent that owns all of a certain commodity (as was the case in the above example), an assumption that is not unusual. The result is easily extended to an arbitrary number of agents. We present here the case $n=2$. 

\begin{proposition}
Consider a Cobb--Douglas economy with $n = 2$ players and $m = 2$ commodities, in which both players hold a strictly positive amount of both commodities. The incentive ratio is $e^{1/e}$ and this bound is tight.
\end{proposition}

\begin{proof}
First consider such an exchange economy in its most general form: 
\begin{itemize}
\item Endowments: $e_1 = (e_{11}, e_{12}), e_2 = (1-e_{11}, 1-e_{12})$, $(e_{11}, e_{12}) \in (0,1)^2$
\item Utility functions: $u_1(x_{11}, x_{12}) = {x_{11}}^\alpha {x_{12}}^{1-\alpha}$, $u'_1(x'_{11}, x'_{12}) = {x'_{11}}^{\alpha'} {x'_{12}}^{1-\alpha'}$ and $u_2(x_{21}, x_{22}) = {x_{21}}^\beta {x_{22}}^{1-\beta}$; $\alpha, \alpha', \beta \in (0,1)$. 
\end{itemize} 

Notice that, here both agents have a little bit of both commodities so we can not be in the situation of the example above. Normalize $p_1 = 1$. Then the value of agent 1's endowment is $p \cdot e_1 = e_{11} + e_{12}p_2$ and  the value of agent 2's endowment  is $p \cdot e_2 = 1-e_{11} + (1-e_{12})p_2$. Demands follow from the Karush-Kuhn-Tucker conditions as per usual: 

\begin{equation*}
\left\{\begin{array}{lr}
   \displaystyle
   x_{11} = \alpha e_{11} + \alpha e_{12}p_2 \\
   \displaystyle
  x_{12} = (1-\alpha)e_{11}/p_2 + (1-\alpha)e_{12}
        \\
   x_{21}  = \beta(1-e_{11})+\beta(1-e_{12})p_2 \\
   \displaystyle
  x_{22} = (1-\beta)(1-e_{11})/p_2 + (1-\beta)(1-e_{12})
        \end{array}\right.
\end{equation*}

We can solve for $p_2$ and we find that:
$$p_2 = \frac{1-\alpha e_{11} - \beta (1-e_{11})}{\alpha e_{12} + \beta (1-e_{12})}$$
Therefore
$$x_{11} = \frac{\alpha[(1-\beta)e_{12} + \beta e_{11}]}{\alpha e_{12} + \beta e_{22}} \mbox{ and } x_{12} = \frac{(1-\alpha)[(1-\beta)e_{12} + \beta e_{11}]}{1-\alpha e_{11}-\beta e_{21}}$$
Now we can find a closed form expression for the incentive ratio as follows:
$$\frac{u_1(x'_1)}{u_1(x_1)} = \left(\frac{\alpha'[\alpha e_{12} + \beta e_{22}]}{\alpha [\alpha' e_{12} + \beta e_{22}]}\right)^\alpha \left(\frac{(1-\alpha')[1-\alpha e_{11} - \beta e_{21}]}{(1-\alpha)[1-\alpha'e_{11}-\beta e_{21}]}\right)^{1-\alpha} =: {T_1}^\alpha {T_2}^{1-\alpha}$$
The following facts are easily verified:
\begin{enumerate}
\item $\alpha' \geq \alpha \Rightarrow T_1 \leq \alpha'/\alpha$
\item $\alpha' < \alpha \Rightarrow T_1 < 1$
\item $\alpha' \geq \alpha \Rightarrow T_2 \leq 1$
\item $\alpha' < \alpha \Rightarrow T_2 < (1-\alpha')/(1-\alpha)$
\end{enumerate}
These facts and the inequality from \cite{Chen2012} that $\forall x, y \geq 0, x^y \leq e^{xy/e}$ show that the incentive ratio is bounded by $e^{1/e}$. The example in \cite{Chen2012} for Fisher markets shows the bound is tight.
 \end{proof}


This bound can be exceeded, even when $m = 3$, as the following example shows.

\begin{example}[Incentive ratio $> e^{1/e}$]
\label{exa:e1/e}
Suppose the market is as follows:
\begin{equation*}
\left\{\begin{array}{lr}
   \displaystyle
   e_{1}  = (.99, .01, .01) \\
   e_{2}  = (.01, .99, .99) \\ 
   u_1(x_1) = x_{11}^{.2}x_{12}^{.3}x_{13}^{.5} \\
   u_2(x_2) = x_{21}^{.4}x_{12}^{.6} \\
   u'_1(x'_1) = {x'_{11}}^{.85}{x'_{12}}^{.1}{x'_{13}}^{.05} \\
\end{array}\right.
\text{ then }
\left\{\begin{array}{lr}
   \displaystyle
   p = (.398, .597, .201) \\
   x_1 \approx (.202, .202, 1) \\
   u_1(x_1) \approx .4495 \\
   p' = (.4045, .1344, .0201) \\
   x'_1 \approx (.845, .299, 1) \\
   u_1(x'_1) \approx .6731
   \end{array}\right.
\end{equation*}
Therefore the incentive ratio is approximately 1.50; a tight bound remains an open question.
\end{example}

The remainder of this section is devoted to the proof that the incentive ratio for Cobb--Douglas markets is bounded. We will use that, by the AM--GM inequality,

\begin{align*}
\prod_{j=1}^m \left(\frac{\alpha'_{ij}}{\alpha_{ij}}\frac{p' \cdot e_i}{p \cdot e_i}\frac{p_j}{p'_j}\right)^{\alpha_{ij}} 
& \leq \sum_{j=1}^m \frac{p' \cdot e_i}{p \cdot e_i}\max_{\alpha'_i} \frac{\alpha'_{ij}}{p'_j} \max_{\alpha_i} p_j 
\end{align*}
and then try to bound the $j-th$ term in this sum by choosing a particular normalisation for $p, p'$.
The following proposition contains three parts: first, it shows that $p_j$ attains its maximum at $\alpha_i = (0, \dots, 1, \dots, 0)$ where the 1 is the $j-th$ element in the vector; second $\alpha'_{ij}/p'_j$ also attains its maximum at $\alpha'_i = \alpha_i = (0, \dots, 1, \dots, 0)$. Finally, the budgets $p \cdot e_i$ and $p' \cdot e_i$ are equal for any choice of $\alpha_i, \alpha'_i$. 

\begin{proposition}
\label{bounds}
\begin{enumerate}
\item [i)] Normalise $p_j = 1$, then $p_j(\alpha_i)$ reaches its maximum when $\alpha_{ij} = 1$ and $\alpha_{ik} = 0$ for all $k \neq j$. 
\item [ii)] Normalise $p'_j = 1$, then $\alpha'_{ij}/p'_j(\alpha'_i)$ reaches its maximum when $\alpha'_{ij} = 1$ and $\alpha'_{ik} = 0$ for all $k \neq j$. 
\item [iii)] Let $A$ be the ``exponents matrix'' i.e. it contains the strategies for each player in such a way that column $i \in [n]$ contains $\alpha_i$ and so $A$ is an $m \times n$ matrix. Similarly,  $E$ is the ``endowment matrix'', where columns are indexed by agents and rows by commodities, so that it is $m \times n$.  Consider the matrices $EA^T-I_m$ and $E(A')^T-I_m$ and their adjugates, $Adj(EA^T-I_m)$ and $Adj(E(A')^T-I_m)$ respectively. Then the first row of the adjugate matrices contains the equilibrium price vectors $p$ and $p'$ (upto a nonzero constant) and moreover, $p \cdot e_i = p' \cdot e_i$.
\end{enumerate} 
\end{proposition}

\begin{proof}
For the first two points in the proof, without loss of generality we consider the case where $j = m$ i.e. $p_m = 1$. \\
$i)$ Let $p(\alpha)$ be the equilibrium price when players report strategies according to $\alpha = (\alpha_1, \dots, \alpha_n)$. We can see a strategy as follows: $\alpha_{ij} = \alpha_{ij} / \sum_{k=1}^m \alpha_{ik}$ for all $i \in [n], j \in [m]$  i.e. we transform the explicit constraint to an implicit one. 
Next, let $\hat{p}$ and $\hat{\alpha_i}$ be vectors with (the first) $m-1$ components and $\hat{z}$ is the excess demand function for the first $m-1$ commodities.
By the implicit function theorem 
\begin{equation*}
D_{\hat{\alpha}_{im}} \hat{p}(\alpha) = - \left[ D_{\hat{p}} \hat{z}(p(\alpha); \alpha)\right]^{-1}D_{\hat{\alpha}_{im}} \hat{z}(p(\alpha); \alpha)
\end{equation*}
By Proposition 17.G.3 in \cite{MWG}, $\left[ D_{\hat{p}} \hat{z}(p(s); \alpha)\right]^{-1}$ has all its entries negative. \\ $D_{\hat{\alpha}_{im}} \hat{z}(p(\alpha); \alpha)$ is a negative vector, this completes the proof.  

$ii)$ Immediate.

$iii)$ We now choose a particular normalisation for $p$ and show that we have $p \cdot e_i = p' \cdot e_i$.
 We apply an argument along the lines of \cite{Press26062012}. First, notice that an equilibrium price vector $p$ satisfies (see \cite{CurtisEaves1985})

$$p^T (EA^T - I_m) = 0$$

We have, where $e_j \alpha_k$ is a shorthand notation for $\sum_{i = 1}^n e_{ij} \alpha_{ik}$,

$$
EA^T - I_m = 
\begin{bmatrix}
e_1\alpha_1 -1 & e_1\alpha_2 & \hdots & e_1 \alpha_m \\
e_2\alpha_1  & e_2\alpha_2-1 & \hdots & e_2 \alpha_m \\
\vdots & \vdots & \ddots & \vdots \\
e_m \alpha_1 & e_m \alpha_2 & \hdots & e_m \alpha_m-1
\end{bmatrix}
$$

As argued in \cite{Press26062012}, the rows of Adj($EA^T - I_m$) are proportional to $p$. We only need to make sure that the ``proportionality factor'' $c$ is not 0. The $(1,j)-th$ entry of Adj($EA^T - I_m$) is what we get when we compute  the determinant of the matrix that we get when removing the first column and $j-th$ row from $EA^T - I_m$ if $j$ is odd and $-1$ times this determinant when $j$ is even. For example, for $p_1$, because of the assumption, 

\begin{equation*}
\left|
\begin{matrix}
e_2\alpha_2-1 & e_2\alpha_3 & \hdots & e_2 \alpha_m \\
e_3\alpha_2 & e_3\alpha_3-1 & \hdots & e_3 \alpha_m \\
\vdots & \vdots & \ddots & \vdots \\
e_m\alpha_2 & e_m \alpha_3 & \hdots & e_m \alpha_m-1
\end{matrix}\right|
\end{equation*}

is a (strictly) diagonal matrix and therefore, nonsingular by the Levy--Desplanques theorem. This shows that the above determinant is nonzero and so, up to the sign, the equilibrium price of commodity 1, $p_1$. Therefore $c \neq 0$ and we conclude that up to a (possibly negative) scalar, we have the equilibrium price $p$ in the first row of Adj($EA^T - I_m$). 

We note that, instead of computing equilibrium prices, we could just compute $p \cdot e_i$ directly by replacing the first column in $EA^T - I_m$ by $e_i$ and computing the determinant (with respect to the first column) of the resulting matrix: 

\begin{equation*}
p \cdot e_i = \frac{1}{c} \left|
\begin{matrix}
e_{i1} & e_1\alpha_2 & \hdots & e_1 \alpha_m \\
e_{i2}  & e_2\alpha_2-1 & \hdots & e_2 \alpha_m \\
\vdots & \vdots & \ddots & \vdots \\
e_{im} & e_m \alpha_2 & \hdots & e_m \alpha_m-1
\end{matrix}\right|
\end{equation*}


We now prove the claim by noting that $p \cdot e_i = p' \cdot e_i$ if and only if this property is true when we change just two entries in $\alpha_i$ and leave the others unchanged i.e. $\alpha'_j = \alpha_j + \delta$ and $\alpha'_k = \alpha_k - \delta$ for $0 \leq \delta \leq \alpha_k$ for $j, k \in [m], j \neq k$ and $\alpha'_r = \alpha_r$ for all $r \in [m] \backslash \{j,k\}$.

We distinguish two cases depending on whether $1 \in \{j, k\}$. First suppose that $j = 1$ and without loss of generality assume $k = 2$. Let $\boldsymbol\delta$ be the vector with zeroes except at position $i$ where it equals $\delta$. Then we have
\begin{align*}
p' \cdot e_i = & \frac{1}{c}  \left|
\begin{matrix}
e_{i1} & e_1(\alpha_2 -\boldsymbol\delta) & \hdots & e_1 \alpha_m \\
e_{i2} & e_2(\alpha_2 -\boldsymbol\delta) -1 & \hdots & e_2 \alpha_m \\
\vdots & \vdots & \ddots & \vdots \\
e_{im} & e_m (\alpha_2-\boldsymbol\delta) & \hdots & e_m \alpha_m-1
\end{matrix}\right| \\
= & \frac{1}{c} \left( \left|
\begin{matrix}
e_{i1} & e_1\alpha_2 & \hdots & e_1 \alpha_m \\
e_{i2}  & e_2\alpha_2-1 & \hdots & e_2 \alpha_m \\
\vdots & \vdots & \ddots & \vdots \\
e_{im} & e_m \alpha_2 & \hdots & e_m \alpha_m-1
\end{matrix}\right| 
 +\left|
\begin{matrix}
e_{i1} & -\delta e_{i1} & \hdots & e_1 \alpha_m \\
e_{i2}  & -\delta e_{i2} & \hdots & e_2 \alpha_m \\
\vdots & \vdots & \ddots & \vdots \\
e_{im} & -\delta e_{im} & \hdots & e_m \alpha_m-1
\end{matrix}\right|\right) \\
= &  p \cdot e_i
\end{align*}
and the conclusion follows; the case where $j, k \neq 1$ is similar: adding up columns $j, k$ in the expression for the determinant does not change the value of $p \cdot e_i$ and $p' \cdot e_i$ and after that we can compute the determinant by splitting it as above.
\end{proof}

\begin{theorem}
The incentive ratio for Cobb--Douglas markets is at most $m$.
\end{theorem}

\begin{proof}
By the results presented in Proposition \ref{bounds}, 
\begin{align*}
\frac{u_i(x'_i)}{u_i(x_i)} = \prod_{j=1}^m \left(\frac{\alpha'_{ij}}{\alpha_{ij}}\frac{p' \cdot e_i}{p \cdot e_i}\frac{p_j}{p'_j}\right)^{\alpha_{ij}} & \leq \prod_{j=1}^m \left(\frac{1}{\alpha_{ij}}\frac{p' \cdot e_i}{p \cdot e_i}\max_{\alpha'_i} \frac{\alpha'_{ij}}{p'_j} \max_{\alpha_i} p_j \right)^{\alpha_{ij}} \\ 
&  \leq \prod_{j=1}^m \left(\frac{1}{\alpha_{ij}}\right)^{\alpha_{ij}}  \leq m,
\end{align*}
where the last step follows from the weighted AM--GM inequality.
\end{proof}

Before concluding, we summarize the results presented here in Table \ref{tab:overview}.
\begin{table}[!h]
\caption{Upper bounds on the incentive ratio in $n \times m$ exchange economies. The bounds for Fisher markets come from \cite{Chen2011}, \cite{Chen2012}.}
\begin{center}
\begin{tabular}{ l c  c }
    \hline
    Market & Fisher & Exchange \\
    \hline
    Leontief & 2 & $\infty$  \\ 
       Linear & 2 & $\infty$ \\ 
    Cobb--Douglas (Without Assumption \ref{assumption1}, $n = 2$, $m =2$) & $e^{1/e}$ & $\infty$\\
    Cobb--Douglas (With Assumption \ref{assumption1}, any $n$, $m =2$) & $e^{1/e}$ & $e^{1/e}$\\
    Cobb--Douglas (With Assumption \ref{assumption1}, any $n$, $m \geq 3$) & $e^{1/e}$ & $m$ \\
    \hline
\end{tabular}
\end{center}
\label{tab:overview}
\end{table}

\section{Conclusion}
\label{sec:conclusion}
This paper surveyed the concept of incentive ratio in the much more general model of exchange economies as compared to Fisher markets. Results in Fisher markets were encouraging: the maximum gains from strategic behaviour were bounded by reasonably small constants and therefore equilibrium mechanisms could be expected to work rather well, meaning that the profits from (computationally challenging) strategic behaviour were small relative to the costs, and thus, not worthwhile on most occassions. In other words, the equilibrium mechanism in Fisher markets is quite robust against strategic behaviour. However, the results here indicate that in the more general setup of exchange economies, results are diametrically different and without further restrictions, ratios in linear, Leontief and Cobb--Douglas markets are unbounded. 

\bibliographystyle{abbrv}
\bibliography{bibliography}

\end{document}